\documentclass[11pt,a4paper,reqno]{amsart}

\usepackage{amsmath,amssymb,graphics,epsfig,color,enumerate,psfrag}
\usepackage{dsfont}  
\usepackage{verbatim}
\newtheorem{Theorem}{Theorem}[section]
\newtheorem{Fact}{Fact}[section]

\newtheorem{Proposition}[Theorem]{Proposition}
\newtheorem{Corollary}[Theorem]{Corollary}
\newtheorem{Remark}[Theorem]{Remark}

\newtheorem{Definition}[Theorem]{Definition}

 \definecolor{darkgreen}{rgb}{0,0.4,0}

\definecolor{light}{gray}{.9}

\newcommand{\cA}{\ensuremath{\mathcal A}}

\newcommand{\cC}{\ensuremath{\mathcal C}}

\newcommand{\cE}{\ensuremath{\mathcal E}}

\newcommand{\cI}{\ensuremath{\mathcal I}}

\newcommand{\cK}{\ensuremath{\mathcal K}}

\newcommand{\cP}{\ensuremath{\mathcal P}}
\newcommand{\cQ}{\ensuremath{\mathcal Q}}

\newcommand{\cT}{\ensuremath{\mathcal T}}

\newcommand{\bbP}{{\ensuremath{\mathbb P}} }

\newcommand{\bbR}{{\ensuremath{\mathbb R}} }

\newcommand{\bbX}{{\ensuremath{\mathbb X}} }
\newcommand{\bbY}{{\ensuremath{\mathbb Y}} }

\let\a=\alpha    \let\d=\delta  
 \let\g=\gamma       \let\l=\lambda
            
   \let\t=\tau   \let\th=\vartheta
  
     \let\L=\Lambda

\begin{document}

\begin{abstract}   In this short note we consider semi--Markov processes satisfying the condition of direction--time independence (Markov renewal processes). We derive large deviation principles and fluctuation theorems for the empirical current and the empirical currents along cycles. Our  derivation  is based on the   joint 
LDP for the empirical measure and flow recently proved in \cite{MZ}.

\medskip


\end{abstract}

\author[A.\ Faggionato]{Alessandra Faggionato}
\address{Alessandra Faggionato.
  Dipartimento di Matematica, Universit\`a di Roma `La Sapienza'
  P.le Aldo Moro 2, 00185 Roma, Italy}
\email{faggiona@mat.uniroma1.it}

\title[LDP's and FT's  for currents in  semi--Markov  processes]{Large deviation principles and fluctuation theorems for currents in  semi--Markov  processes}

\maketitle

\section{Introduction}
Semi--Markov processes with direction--time independence are stochastic processes similar to continuous time Markov chains with the exception that the holding times  are  not necessarily exponential random variables (hence, these processes are in general non Markovian). In the mathematical literature they  are also  known as  \emph{Markov renewal processes}  \cite{As}.  They find several applications, also in the study of molecular motors (cf. e.g. \cite{As,AG3,FD,FS2,L,MNW,MZ,WQ} and references therein). 

For several  Markov processes (as diffusions and  Markov chains) in the last years much attention has been devoted to  the large deviations of the empirical current and the associated fluctuation theorems (also called Gallavotti--Cohen symmetries). See e.g. 
\cite{BFG1,BFG2,LS} and references therein.

Considering semi--Markov processes with direction--time independence, previous derivations of large deviation principles for the  joint empirical measure and current as well for the empirical current have been obtained (in a not completely rigorous way) in \cite{AG3,MNW}. Fluctuation theorems have also been discussed in particular in  \cite{AG3}, also for empirical  currents along cycles. 

In this short note we show how the above  LDPs and the fluctuation theorems can be derived from the  joint 
LDP for the empirical measure and flow recently proved in \cite{MZ}.
We also give some extension to  generic semi--Markov processes (without  direction--time independence). Our derivation  covers also  the  case of semi--Markov processes with  holding times having   law with heavy tails or without a probability density (these cases indeed do not fit well with the arguments presented in \cite{AG3,MNW}). In addition, our  derivation is given by simple   mathematical proofs.

\section{Semi--Markov processes}

\subsection{Semi-Markov processes with direction--time independence (DTI)} 
 Given a  finite state space $V$, the DTI semi--Markov process 
 $\bbX:=(\bbX _t)_{ t \geq 0}$  on $V$  is defined from the following objects: a  transition probability kernel  $(p_{x,y})_{x,y \in V} $  that we assume to be irreducible, a  probability measure $\g$ on $V$ and  a family of probability measures $\psi_x $ on $(0,+\infty)$ parametrized by $x \in V$. 
Having  these objects, we introduce 
  a discrete--time Markov chain $(X_k,\t_k)_{k \geq 0}$ on $V \times (0,+\infty)$ such that 

\begin{itemize}

\item[(C1)]
$(X_k) _{k \geq 0}$ is a Markov chain  on $V$ with transition probabilities $p_{x,y}$, $x,y \in V$,   and initial  distribution $\g$. 
 By the above assumption, this Markov chain  is irreducible;

\item[(C2)]  $(\t_k)_{k \geq 0} $ is a random sequence on $(0, +\infty)$ such that, conditionally to $(X_k)_{k \geq 0}$,  $(\t_k)_{k \geq 0} $ are i.i.d random  variables  and   $\t_k$ has law $\psi _{X_k}$, i.e. 
\begin{equation}\label{uffina}
\bbP_\g \bigl( \t_i \in A \,|\, ( X_k)_{k \geq 0} \bigr) = \psi_{X_i } (A)  \qquad  \qquad i \geq 0\,, \;\; A \subset \bbR \text{ measurable} \,.
\end{equation}
\end{itemize}

\medskip

Then the DTI semi--Markov process     $\bbX= (\bbX_t)_{t \geq 0}$ is  obtained  from $(X_k)_{k \geq 0}$   by the following random  time--change: 
 at time $0$ the system starts at state $X_0$ and it   remains there for a holding  time  $\t_0$, at time $\t_0$ the system  jumps to state $X_1$ and it remains  there for a holding time $\t_1$ and so on.
 We can formalize this definition as follows. 
We set 
\begin{equation}\label{lazio}S_0:=0\,, \qquad \qquad S_k:= \sum_{i=0}^{k-1} \t_i  \;\; \text{ for }  k \geq 1\,.
\end{equation}
 Then, given $ t \geq 0$, we define $N_t$ as the unique nonnegative integer $k$  such that $ S_k \leq t < S_{k+1}$. 
 Note that the above definition is well posed $\bbP_\g$--a.s. since, as one can easily prove, $\bbP_\g$--a.s. it holds $ \lim _{k \to \infty} S_k = +\infty$. Then we define
 \begin{equation}\label{marche}
 \bbX_t:= X_{N_t} \,.
 \end{equation}
Note that $X_t= X_0$ for $t \in [0,\t_0)=[S_0,S_1)$, $X_t=X_1$ for $t\in [\t_0,\t_0+\t_1)=[S_1,S_2)$,...

\smallskip 

When  $\psi_x (dt ) $ is of the form $ f(x,t) dt $ for some density function $f(x, \cdot)$, then  $\bbX_t$ corresponds to the process introduced in \cite[Section 2.1]{MNW} with $Q(x,t) = f(x,t)$ there. Note that, when $f(x, t) = \l_x e^{- \l_x t}$, then $\bbX_t$ is simply a continuous--time  Markov chain on $V$ with transition probability rates $r_{x,y}= \l_x p_{x,y}$.

\smallskip

\subsection{Generic semi-Markov processes}\label{spada}
The condition of  direction--time independence  corresponds to the fact that the law of  $\t_i$ is determined when $X_i$ is known. In a generic semi--Markov process the law of $\t_i$ is determined when $X_{i},X_{i+1}$ are known, in particular  the holding time  at $X_i$ can depend   also from the  new state $X_{i+1}$ achieved after the transition. 

As a consequence, instead of working with the family $\{ \psi_x \}_{x \in V}$, we have  a family of probability measures $\psi_{x,y}$ on $(0,+\infty)$ parameterized by $(x,y ) \in V \times V$. 
Then  one again consider the discrete--time process  $(X_k, \t_k) _{k \geq 0}$ (which is not anymore Markov) satisfying conditions (C1) and $(\text{C2}^{\star})$, where the new condition 
$(\text{C2}^{\star})$ reads as follows:

\begin{itemize}

\item[$(\text{C2}^{\star})$] $(\t_k)_{k \geq 0} $ is a random sequence on $(0, +\infty)$ such that, conditionally to $(X_k)_{k \geq  0 }$,  $(\t_k)_{k \geq 0} $ are i.i.d random  variables  and   $\t_k$ has law $\psi _{X_{k},X_{k+1}}$, i.e. 
\begin{equation}\label{uffina_bis}
\bbP_\g \bigl( \t_i \in A \,|\, ( X_k)_{k \geq 0} \bigr) = \psi_{X_{i}, X_{i+1} } (A) \qquad \forall i \geq 0,, \;\; A \subset \bbR \text{ measurable} \,.
\end{equation}
\end{itemize}
Then the semi--Markov process $\bbX_t$ is again defined by \eqref{lazio} and  \eqref{marche}.

\medskip

For the above definition it is simple to check that 
\begin{equation}\label{saurus}
\bbP_\g \bigl( \t_k \in A \,, \; X_{k+1} =y\, |\, 
X_0,X_1, \dots, X_k, \t_0, \t_1, \dots, \t_{k-1} \bigr)=p(X_k,y) \psi _{X_k,y} (A)\,.
\end{equation}
In particular, the above defined semi--Markov process  corresponds to the one introduced in \cite[Appendix A.1]{MNW} by setting there  $\L(x,y;t):=p_{x,y} \psi_{x,y} \bigl ( \,(t, +\infty)\bigr)$. When $\psi_{x,y} $ is absolutely continuous w.r.t. the Lebesgue measure on $(0,+\infty)$, i.e. $\psi_{x,y}(dt) = f_{x,y} (t) dt $, we then obtain  that the function  $Q(x,y;t)$  in \cite[Appendix A.1]{MNW} equals $p_{x,y} f_{x,y}(t)$.

\medskip

We recall that any generic semi--Markov process on $V$ with irreducible transition kernel can be formulated in terms of  a DTI  semi--Markov process on $E:=\{ (x,y) \in V \times V\,:\, p(x,y)>0 \}$ with irreducible transition kernel. To this aim,  
 consider the discrete--time  Markov chain $(Y_k)_{k \geq 0}$ on $E$, 
 with   irreducible transition kernel given by $\hat p_{(x,y) , (v,z) }= \d_{y,v} p_{y,z}$ and initial distribution    given by the  distribution of $(X_0, X_1)$ under $\bbP_\g$.  We  write  $( \bbY_t)_{t\geq 0}$ for the associated semi--Markov process with 
  direction--time independence
  such that $\psi_{x,y}$ is the holding time distribution  at state $(x,y) \in E$.   Then the semi--Markov process $(\bbX_t)_{t \geq 0}$ can be realized simply by defining $\bbX_t$ as the first coordinate of $\bbY_t$.

%

%
%

\subsection{Empirical measure and flow}
 Given $t>0$, the \emph{empirical measure}
$\mu_t$  is defined as the random probability measure  
\begin{equation}
\mu_t: = \frac{1}{t} \int_0^t \d_{(X_{N_s}, \t_{N_s} )} ds=  \frac{1}{t} \int_0^t \d_{(\bbX_s, \t_{N_s} )} ds
\,.
\end{equation}
In other words, the expectation $\mu_t(f)$ of a function $f$  is given by  
\[
\mu_t (f)= \frac{1}{t} \int_0^t f( \bbX_s, \t_{N_s} ) ds  = \sum _{k =0}^{N_t-1} \frac{\t_k }{t}  f( X_k, \t_k) + \frac{t-S_{N_t} }{t} f(X_{N_t}, \t_{N_t} ) \,.
\] 
The empirical measure $\mu_t$ is a probability on $ V \times (0,+\infty)$ and, by trivial extension,   can be thought of as an element of  $\cP( V \times (0,+\infty])$, the space of probabilities on $ V \times (0,+\infty]$.

The \emph{empirical flow} is defined as the random element of $\bbR_+^{V \times V}$ given by 
\begin{equation}\label{piero}
Q_t(x,y):=\frac{1}{t} \sum_{k=0}^{N_t-1}\mathds{1} ( X_k=x\,, X_{k+1} = y )\,.\end{equation}
We recall that $\bbR_+:=[0,+\infty)$.
 Note that, for $x\not = y$,  we have 
 $Q_t(x,y)= \frac{1}{t} \sum_{s \in (0, t]} \mathds{1}( \bbX _{s-}=x, \bbX _{s}=y)$. If one allows  $p_{x,x}$ to be positive, then $Q_t(x,x)$ can be positive.

 \smallskip
 
 Due to the above definitions, the joint empirical measure and flow $( \mu_t, Q_t) $ is a random  element of the 
  the product space
 \begin{equation}
  \L := \cP( V \times (0,+\infty]) \times \bbR_+ ^{V \times  V}\,.
  \end{equation}


\section{LDP for the joint  empirical measure and flow  for DTI semi--Markov processes  \cite{MZ}}\label{suntoMZ}

In this section we restrict to DTI semi--Markov processes and 
we recall  the joint large deviation principle  for the empirical measure and flow recently obtained by Mariani and Zambotti \cite{MZ}.  We point out that our notation is slightly different  from the one in \cite{MZ} since they call $\t_{k+1}$ our random variable $\t_k$.

\smallskip

We write $\cP( V \times (0,+\infty])$ for the space of probability measures on  $V \times (0,+\infty]$ ($V$ has the discrete topology, and $(0,+\infty]$ is a metric space by the identification
 $(0,+\infty] \ni x \mapsto  \frac{x}{1+x}\in  (0,1]$). The space $ \cP( V \times (0,+\infty])$  is endowed with the weak topology. We also consider the euclidean  space $\bbR_+^{V \times V}$ of functions  $ V \times V \ni (x,y) \mapsto Q(x,y) \in \bbR_+$.
 
 \smallskip
 
 Let us write $\nu$ for the unique invariant distribution of the Markov chain $(X_k)_{k\geq 0}$. As discussed in \cite[Section 4]{MZ}, as $t\to \infty$ the empirical measure $\mu_t$ satisfies the following  LLN  for any initial distribution $\g$:
 \[\mu_t(x, d\t)  \to  \frac{ \nu_x}{\sum _z \nu_z \int \t \psi_z(d \t) } \t\psi_x (d\t)\,, \qquad \bbP_{\g}\text{--a.s.}
 \] 
Again,  in \cite[Section 4]{MZ}, it is proved that as $t\to \infty$ the empirical flow $Q_t$ satisfies the following  LLN  for any initial distribution $\g$:
 \[Q_t (x,y) \to  \frac{ \nu_x p_{x,y} }{\sum _z  \nu_z\int \t \psi_z(d \t) } \,, \qquad \bbP_{\g}\text{--a.s.}
 \] 
To descrive the large deviations  from the above LLN's  we need  some notation.

 \begin{Definition}\label{chirurgia}
 We define $\L_0$ as the subspace of $\L$ given by the  pairs $(\mu, Q) $ such that, for any $x \in V$,  the following holds:
\begin{itemize}
\item the measure  $ \mu (x, d\t)$ restricted to $(0,+\infty)$ is absolutely continuous w.r.t. $\psi_x$,
\item  $ Z_x:=  \int_{(0, +\infty)} \mu(x, d\t)  \frac{1}{\t}= \sum _{y\in V} Q(x,y) = \sum _{y\in V} Q(y,x)$.
\end{itemize}
Given  $(\mu,Q) \in \L_0$, we define 
\begin{equation}\label{arrosticino1}  p_{x,y}^Q := \frac{Q(x,y) }{Z_x} \,, \qquad \tilde \mu (x,d\t) =\frac{1}{Z_x} \frac{1}{\t} \mu(x, d\tau) \,,
\end{equation}
with the convention that $ \tilde \mu (x,\{ +\infty \} )=0$.\end{Definition}

Note that    $p_{x,y}^Q $ is a probability kernel  on $V$ and that $\tilde \mu$ is a probability measure on $V \times (0,+\infty)$.
We also  point out that if $(\mu,Q) \in \L_0$ then ${\rm div}\, Q\equiv 0$. We recall that,
 given an element $\mathfrak{f} \in \bbR^{V \times V}$, the divergence ${\rm div}\,\mathfrak{f}: V \to \bbR$ is defined as
 \begin{equation}\label{def_div} {\rm div}\, \mathfrak{f}(x) := \sum _{y \in V} \mathfrak{f}(x,y)- \sum_{y \in V} \mathfrak{f}(y,x) \,.
 \end{equation}
 In the case of a flow $Q$, the divergence ${\rm div}\, Q(x)$ is simply the difference between  the flow exiting from $x$ and the flow entering into $x$.

In what follows,   given two probability measures $P,P'$, we denote by $H(P|P')$ the entropy of $P$ w.r.t. $P'$.

 \begin{Fact}\label{vettore} (Joint LDP for $(\mu_t, Q_t)$  by Mariani \& Zambotti   
 \cite{MZ})
 
  Under $\bbP_\g$ the random pair $(\mu_t, Q_t)$ satisfies a large deviation principle  as $t \to \infty$, with speed $t$ and explicit rate function $I$ given by
  \begin{equation} \label{eq:LDP}
  I(\mu,Q):= 
  \begin{cases}
  \sum _{x \in V} Z_x \left[ H \bigl( p_{x, \cdot} ^Q | p_{x, \cdot} \bigr ) +H \bigl(
  \tilde\mu (x, \cdot) | \psi_x \bigr)   + \xi _x \mu (x , \{+\infty\}) \right] & \text{ if } (\mu, Q) \in \L_0\,,\\
  +\infty & \text{ otherwise}\,,
  \end{cases} 
  \end{equation}
  where \footnote{We use the convention that $\xi _x \mu(x, \{+\infty\})=0$ if $\xi_x=\infty$ and $\mu(x, \{+\infty\})=0$.}
  \begin{equation}\label{arrosticino2}
\xi_x := \sup \{ c \geq 0 \,:\, \int \psi_x (d\t) e^{c \t} < \infty\}\,, \qquad x \in V \,.
\end{equation}
  Moreover, the rate function $I$ is good, i.e. the level set $\{ (\mu, Q) \,:\, I(\mu,Q) \leq \a \}$ is compact for any $ \a \in [0, +\infty)$.
 \end{Fact}

 Note that if
 e.g.  $\psi_x$ is an exponential distribution with parameter   $\l_x$, then $\xi(x) = \l_x$.

 \section{Fluctuation theorem for the empirical  current of DTI semi--Markov processes }\label{FT_for_currents}
 
 We denote by $\bbR^{V \times V} _{\text{antis.} }$ the space of antisymmetric functions $J : V \times V \to \bbR$ (equivalently, antisymmetric real square  matrixes with indexes in $V$).
 The \emph{empirical current} is defined as the random element of   $\bbR^{V \times V}_{\text{antis.}}$ given by 
\begin{equation}
J_t(x,y):= Q_t(x,y)-Q_t(y,x)\,, \qquad (x,y) \in V \times V\,,
\end{equation}
 i.e., for $x \not  =y$,  $J_t(x,y)$ is given by the number of transitions   per unit time from $x$ to $y$ minus the  number of transitions   per unit time from $y$ to $x$ performed by the semi--Markov process $(\bbX_s) _{s \in [0,t]}$. Trivially, $J_t(x,x)=0$.

\medskip

Given $Q \in \bbR_+^{V \times V}$ we define $J ^Q \in \bbR^{V \times V}_{\text{antis.}}$  as 
\begin{equation} \label{albero}
 J^Q(x,y) = Q(x,y) - Q(y,x) \,, \qquad (x,y) \in V \times V \,.
\end{equation}
By applying the contraction principle to Fact \ref{vettore} we get:

\begin{Proposition}\label{tititi}$[$LDP for $(\mu_t, J_t)$ and LDP for $J_t]$ Under $\bbP_\g$,  the random pair $(\mu_t, J_t)$ satisfies a large deviation principle  as $t \to \infty$, with speed $t$ and  good rate function $\tilde I $ given by
 \begin{equation}\label{altino}
 \tilde I (\mu, J)= \inf \{ I(\mu,Q) \,:\, Q \in \bbR_+^{V \times V}  \,, \; J^Q=J\}\,.
 \end{equation}
 Similarly, 
 under $\bbP_\g$,  the empirical current  $J_t$ satisfies a large deviation principle  as $t \to \infty$, with speed $t$ and  good rate function $\cI  $ given by
 \begin{equation}\label{altino_bis}
 \cI ( J)= \inf \{ I(\mu,Q) \,:\, (\mu,Q) \in \L \,, \; J^Q=J\}\,.
 \end{equation}
\end{Proposition}

We now move to the fluctuation theorem.  To this aim, as usual, we restrict to the case
\begin{equation}\label{simone}
p_{x,y}>0 \text{ if and only if } p_{y,x}>0\,.
\end{equation}
It is convenient to introduce the set $E$ of pairs $(x,y)$ such that both the transition from $x$ to $y$ and the transition from $y$ to $x$ are possible:
\begin{equation}\label{simone_bis}
 E:= \{ (x,y) \in V \times V\,:\, p_{x,y}>0\,, \; p_{y,x}>0\}\,.
\end{equation}

As we will show, the fluctuation theorem follows from the following key symmetry of the rate functional $I(\mu,Q)$:
\begin{Theorem}\label{chiave}
For any $(\mu,Q) \in \L$ it holds
\begin{equation}\label{treno}
I(\mu,Q) = I(\mu, Q^T) -\frac{1}{2} \sum _{(x,y) \in E} J^Q(x,y) \ln \frac{ p_{x,y}}{ p_{y,x} }\,,
\end{equation}
where $Q^T(x,y):= Q(y,x)$ and $J^Q$ is given by \eqref{albero}.
\end{Theorem}
Since the rate function $I$ has value in $(-\infty, +\infty]$ while $\sum _{(x,y) \in E} J^Q(x,y) \ln \frac{ p_{x,y}}{ p_{y,x} }$ is finite,   the above identity is well defined in $(-\infty,+\infty]$.
 \begin{proof} To simplify notation we write $J$ instead of $J^Q$. 
Trivially, $(\mu,Q) \in \L_0$ if and only if $(\mu, Q^T) \in \L_0$ (cf. Definition \ref{chirurgia}). 
If $(\mu,Q) \not \in \L_0$,  $(\mu, Q^T) \not \in \L_0$, then \eqref{treno} reads $+\infty= +\infty$, which is trivially true. Hence we can restrict to the case $(\mu,Q) \in \L_0$, $(\mu, Q^T) \in \L_0$. Due to \eqref{eq:LDP}, to prove \eqref{treno} it is enough to show that
\begin{equation}\label{aereo}
  \sum _{x \in V} Z_x  H \bigl( p_{x, \cdot} ^Q | p_{x, \cdot} \bigr ) =
   \sum _{x \in V} Z_x  H \bigl( p_{x, \cdot} ^{Q^T} | p_{x, \cdot} \bigr )  - \frac{1}{2} \sum _{(x,y) \in E} J(x,y) \ln \frac{ p_{x,y}}{ p_{y,x} }\,.
    \end{equation}
  We point out that $Z_x >0$ for any $x$.
 Hence, the l.h.s. of \eqref{aereo} is infinite if and only  if the following condition $\cC$ is satisfied:  there exists a pair $(x,y)$ with 
 $p_{x,y}=0$ and $Q(x,y)>0$. On the other hand, the l.h.s. of \eqref{aereo} is infinite if and only if for some $x$ the probability  $p_{x,\cdot}^Q$ is not absolutely continuous w.r.t. $p_{x,\cdot}$, i.e. if and only if there exists a pair $(x,y)$ with 
 $p_{x,y}=0$ and $Q(x,y)>0$. 
 Due to \eqref{simone}  $p_{x,y}=0$ if and only if $p_{y,x}=0$. Hence, we can restate condition $\cC$ as follows: there exists a pair $(y,x)$ such that $p_{y,x}=0$ and $Q^T(y,x)=0$.  This property is equivalent to the fact that $  \sum _{x \in V} Z_x  H \bigl( p_{x, \cdot} ^{Q^T} | p_{x, \cdot} \bigr )  =  \sum _{y \in V} Z_y  H \bigl( p_{y, \cdot} ^{Q^T} | p_{y, \cdot} \bigr )  
 $ is infinite. Hence, under condition $\cC$ \eqref{aereo},  reduces to the identity $+\infty=+\infty$
and therefore it is true.

  Let us suppose that condition $\cC$ is not fulfilled. Then, by the above observations,  the three sums in \eqref{aereo} have finite value. Moreover,  if $p(x,y)=0$ then $p(y,x)=0$, $Q(x,y)=0$ and $Q^T(x,y)=0$. Hence, using the convention that $0 \ln 0=0$, we can write
  \begin{equation}\label{alex1}
  \begin{split}
   \sum _{x \in V} Z_x  H \bigl( p_{x, \cdot} ^Q | p_{x, \cdot} \bigr )& = \sum _{(x,y)\in E}
   Q(x,y) \ln \frac{Q(x,y)}{ Z_x p_{x,y} } \\
   & =\frac{1}{2}  \sum _{(x,y)\in E} \left[ 
   Q(x,y) \ln \frac{Q(x,y)}{ Z_x p_{x,y} } +   Q(y,x) \ln \frac{Q(y,x)}{ Z_y p_{y,x} }\right]\\
   & = \frac{1}{2} \sum _{(x,y)\in E} \left[ 
   Q(x,y) \ln Q(x,y)+ Q^T(x,y) \ln Q^T(x,y)\right]  \\
   & - \frac{1}{2} \sum _{(x,y)\in E} \left[ Q(x,y) \ln (Z_x p_{x,y} )+ Q^T(x,y) \ln ( Z_y p_{y,x} ) \right]   \,.
        \end{split}
   \end{equation}
 Since $(Q^T)^T=Q$ a similar expression holds:
 \begin{equation}\label{alex2}
 \begin{split}
   \sum _{x \in V} Z_x  H \bigl( p_{x, \cdot} ^{Q^T} | p_{x, \cdot} \bigr )& =
   \frac{1}{2} \sum _{(x,y)\in E} \left[ 
   Q^T(x,y) \ln Q^T(x,y)+ Q(x,y) \ln Q(x,y)\right]  \\
   & - \frac{1}{2} \sum _{(x,y)\in E} \left[ Q^T(x,y) \ln (Z_x p_{x,y} )+ Q(x,y) \ln ( Z_y p_{y,x} ) \right] \,.
 \end{split}
 \end{equation} 
 By subtracting \eqref{alex2} from \eqref{alex1} and using that $J=Q-Q^T$, we get
 \begin{equation}\label{zero}
 \begin{split} &  \sum _{x \in V} Z_x  H \bigl( p_{x, \cdot} ^Q | p_{x, \cdot} \bigr )-  \sum _{x \in V} Z_x  H \bigl( p_{x, \cdot} ^{Q^T} | p_{x, \cdot} \bigr ) \\
 & \qquad \qquad = - \frac{1}{2} \sum _{(x,y)\in E} \left[ J(x,y) \ln (Z_x p_{x,y} )- J(x,y) \ln ( Z_y p_{y,x} ) \right]\\
 & \qquad \qquad = - \frac{1}{2} \sum _{(x,y) \in E} J(x,y) \ln \frac{ p_{x,y}}{ p_{y,x} }-\frac{1}{2}
  \sum _{(x,y)\in E} J(x,y) [ \ln (Z_x)- \ln (Z_y) ]\,.
 \end{split}
 \end{equation}
 To get \eqref{aereo} we have only to show that $\sum _{(x,y)\in E} J(x,y) [ \ln (Z_x)- \ln (Z_y) ]$. Recall that we are assuming that $(\mu,Q)\in \L_0$ and that condition $\cC$ is not fulfilled. 
As already observed, the latter implies that $J(x,y)=0$ if $(x,y) \in (V \times V) \setminus  E$, hence 
\begin{equation}
\label{uno}\sum _{(x,y)\in E} J(x,y) [ \ln (Z_x)- \ln (Z_y) ]=  
\sum _{(x,y)\in V \times V} J(x,y) [ \ln (Z_x)- \ln (Z_y) ]\,.
\end{equation}
 Since $(\mu,Q)\in \L_0$, 
  as already observed before \eqref{def_div}, ${\rm div}\, Q\equiv 0$. Hence, using also the antisymmetry of $J$, we get 
 \begin{equation}\label{gelato} 
 \begin{split}
 0& = {\rm div}\, Q(x)= \sum _{y\in V} Q(x,y)- \sum_{y\in V} Q(y,x)=\sum_{y\in V} J(x,y)\\
 &  = \frac{1}{2}
 \sum_{y\in V} ( J(x,y)-J(y,x) )= \frac{1}{2} {\rm div}\, J(x) \,,
 \end{split}
 \end{equation} thus proving that ${\rm div}\, J \equiv 0$. Since $J$ is divergenceless, 
 the scalar product of $J$ with a gradient function is zero. In our case, this reads
 \begin{equation}\label{due}
 \begin{split}
  \sum _{(x,y)\in V \times V } J(x,y) [ \ln (Z_x)- \ln (Z_y)] & = \sum _{x\in V} \ln (Z_x) [ \sum_{y\in V} J(x,y)- \sum _{y\in V}  J(y,x) ]\\
  &= \sum_{x\in V} \ln (Z_x) {\rm div}\, J(x)=0\,.  
 \end{split} \end{equation}
 As a byproduct of \eqref{zero}, \eqref{uno} and \eqref{due}, we get \eqref{aereo} and therefore \eqref{treno}.
 \end{proof}
\begin{Remark}\label{letto} Recall the rate function $\tilde I (\mu,Q)$ of Prop. \ref{tititi}.
As derived in the proof of Theorem \ref{chiave}, if $I(\mu,J)<+\infty$ then  ${\rm div} J=0$ and $J(x,y)=0$ for any pair $(x,y) \in V \times V$ such that $p_{x,y}=0$. Moreover, if $I(\mu,Q)<+\infty$, then  ${\rm div} Q=0$ and $Q(x,y)=0$ for any pair $(x,y) \in V \times V$ such that $p_{x,y}=0$
\end{Remark}

We observe that,
given $J \in \bbR^{V \times V}_{\text{antis.}} $,
the map \[
\{ Q \in \bbR_+ ^{V \times V} \,:\, J^Q=J \} \ni \cQ \;\to\; \cQ^T \in \{ Q \in \bbR_+ ^{V \times V} \,:\, J^Q=-J \} 
\]
is bijective.  The above observation, Proposition \ref{tititi} and Theorem \ref{chiave} imply immediately the following fact:
\begin{Theorem}\label{udine}$($Fluctuation theorems for $\tilde I$ and for $\cI$)\\
 The joint LD rate function $\tilde I$ for $(\mu_t,J_t)$ satisfies 
\begin{equation}\label{bu1}
\tilde I(\mu,J)= \tilde I(\mu,-J) -\frac{1}{2} \sum _{(x,y) \in E} J(x,y) \ln \frac{ p_{x,y}}{ p_{y,x} }\,, 
\end{equation}
for any $\mu \in \cP( V \times (0,+\infty])$ and $J\in  \bbR^{V \times V}_{\text{antis.}}$.

Similarly, the LD rate function $\cI$ for $J_t$ satisfies 
\begin{equation}\label{bu2}
\cI(J)= \cI(-J) -\frac{1}{2} \sum _{(x,y) \in E} J(x,y) \ln \frac{ p_{x,y}}{ p_{y,x} }\,, 
\end{equation}
for any $J\in  \bbR^{V \times V}_{\text{antis.}}$.
\end{Theorem}
We recall that the identities \eqref{bu1} and \eqref{bu2} have to be thought in $(-\infty,+\infty]$.


 \section{Fluctuation theorem for the empirical  current along chords}\label{FT_for_chords}
 Considering e.g. applications to biochemical processes (see  e.g. \cite{AG2,AG3,FD,FS2,Sc,WQ} and references therein), it is relevant to extend the above analysis to generalized empirical currents along cycles (or equivalently, chords).

  Again we assume condition \eqref{simone}. Recall \eqref{simone_bis}.
 We consider the unoriented graph $G$ with vertex set $V$ and edges
  \[ \cE = \{ \{x,y\} \,:\, x\not = y, \; (x,y) \in E \} = 
  \{ \{x,y\} \,:\, x\not = y, \; p_{x,y}>0\,,\; p_{y,x}>0 \} \,.
  \]
  Due to our irreducibility assumption on the transition  kernel $p_{x,y}$, the graph $G$ is connected.
  \smallskip
    
  An \emph{oriented cycle} $\cC$ in $G$ is given by a sequence $(z_1, \dots, z_s)$ of vertexes in $V$ such that $(z_i, z_{i+1}) \in E$, with the convention that $z_{s+1}:=z_1$.  To the oriented cycle $\cC$ we associate the affinity $\cA(\cC)$ defined as
  \begin{equation}\label{affine}
  \cA(\cC) = \sum _{i=1}^s 
  \ln \frac{ p_{z_i,z_{i+1} }}{ p_{z_{i+1},z_i }} \,. 
    \end{equation}
 
  \medskip

Fix once and for all an  unoriented  \emph{ spanning tree} $\cT$ in $G$, i.e. a subgraph of the unoriented graph 
$G$ without loops and such that any $x\in V$ is  also a vertex of $\cT$. We recall that the edges of $G$ that do not belong to $\cT$ are called \emph{chords}. For each chord choose once and for all  an orientation, and denote by $\mathfrak{c}_1, \dots, \mathfrak{c}_m$ the oriented chords of $G$. It is known that for each $k=1,\dots,m$ there is a unique self--avoiding oriented  cycle $\cC_k$   starting with the oriented edge $\mathfrak{c}_k$ and lying inside   the graph obtained from $\cT$ by adding the edge $\mathfrak{c}_k$. More precisely, there is a unique cycle 
$\cC_k=(z_1, \dots, z_s)$ such that $z_1,\dots, z_s $ are all distinct vertexes of $V$,  $(z_1,z_2)=\mathfrak{c}_k$ and $(z_i , z_{i+1}) $  is an edge of $\cT$ when disregarding the orientation for  all $i=2,\dots, s$ (with the convention that $z_{s+1}:= z_1$).

\medskip

To each $\cC_k=(z_1, \dots, z_s)$   we associate a special current $J_k \in \bbR^{V \times V}_{\text{antis.}}$ as follows:
\begin{equation}\label{tac}
J_k(x,y) := 
\begin{cases}
1 & \text{ if  $(x,y)=(z_i, z_{i+1})$ for some $i = 1,\dots, s$}\,,\\
-1& \text{ if  $(y,x)=(z_i, z_{i+1})$ for some $i = 1,\dots, s$}\,,\\
0 & \text{ otherwise}\,.
\end{cases}
\end{equation}
Trivially, ${\rm div} J_k=0$ and $J_k (x,y) =0$ if $p_{x,y}=0$ (i.e. if $(x,y) \not \in E$).

\medskip

The following fact is a direct consequence of Lemma 9.3 in \cite{BFG2}:
\begin{Proposition}\label{campana}
Let $J \in \bbR^{V \times V}_{\text{antis.}}$ be such that ${\rm div} J=0$ and $J(e)=0$ for any $ e \not \in E$. Then 
$J= \sum_{k=1}^m J( \mathfrak{c}_k) J_k $.
\end{Proposition}

As consequence only of \eqref{bu2} in Theorem \ref{udine} and the decomposition given in Proposition \ref{campana} we get:

\begin{Theorem}\label{cordine}
Under $\bbP_\g$ the random vector $\bigl (J_t( \mathfrak{c}_1), J_t( \mathfrak{c}_2), \dots, J_t( \mathfrak{c}_m) \bigr)$ satisfies a LDP with speed $t$ and good rate function $\hat{\cI}$ satisfying
\begin{equation}\label{eq:pooh}
\hat{\cI} (\th_1, \dots, \th_m)= \hat{\cI} (-\th_1, \dots, -\th_m)- \sum_{k=1}^m \th_k \cA( \cC_k)\,.
\end{equation}
\end{Theorem}

\begin{Remark}
Having Theorem \ref{cordine} one can easily derived a fluctuation theorem for generalized algebraic currents as in \cite{FD}, i.e. currents associated to a basis $\cC_1, \dots, \cC_m$ of the cycle space where the cycles $\cC_1, \dots, \cC_m$  are  not necessarily built from a spanning tree as above. We refer to \cite{Sc,AG2,BFG2,FD} 
for an overview  on cycle theory, currents along cycles and physical implications, that still hold for  DTI semi--Markov processes due to Theorem \ref{cordine}.
\end{Remark}

\begin{proof}[Proof of Theorem \ref{cordine}]
The map $ \bbR^{V \times V}_{\text{antis.}}  \ni J \mapsto \bigl( J( \mathfrak{c}_1), \dots, J( \mathfrak{c}_m )   \bigr) \in \bbR^m$ is continuous. As a consequence of the contraction principle and the LDP stated in Theorem \ref{udine} we have that, under $\bbP_\g$, the random vector $\bigl (J_t( \mathfrak{c}_1), J_t( \mathfrak{c}_2), \dots, J_t( \mathfrak{c}_m) \bigr)$ satisfies a LDP with speed $t$ and good rate function $\hat{\cI}$ given by
\begin{equation}\label{maestra}
\hat{\cI} (\th_1, \dots, \th_m)= \inf \{ \cI (J) \,:\, J \in  W \} \,,
\end{equation}
where 
\[ W:=
\{ J  \in \bbR^{V \times V}_{\text{antis.}}  \,:\; J(\mathfrak{c}_k)= \th_k \; \forall k=1,\dots,m \}
\,.
\]
Recall the definition of $J_k$ given in \eqref{tac}. We claim that the above infimum in \eqref{maestra} is indeed a minimum attained at $J_*= \sum _{k=1}^m \th_k J_k$, i.e. $\hat{\cI} (\th_1, \dots, \th_m)=\hat{\cI} \bigl( J_*\bigr)$. 
Since $J_i(\mathfrak{c}_k) = \d_{k,i}$, it is simple to check that $J_*$ belongs to $W$. Take now a generic $J \in W$.  Due to Remark
\ref{letto}, $\cI(J) =+\infty$ if $ {\rm div}\, J \not \equiv 0$ or if $J(e)\not =0$ for some $e \not\in E$. On the other hand, by Proposition \ref{campana}, the only element $J \in W$ for which ${\rm div}\, J= 0$ and $J(e) =0$ for all $e \not \in E$ is $J_*$, thus proving our claim.

\medskip
By the previous  observation we also have $ \hat{\cI}(  -\th_1, \dots, -\th_m)=\hat{\cI} \bigl(-J_* \bigr)$. Hence, as a consequence of \eqref{bu2}, we have 
\[  \hat{\cI}(  \th_1, \dots, \th_m)=  \hat{\cI}(  -\th_1, \dots, -\th_m)-\frac{1}{2} \sum _{(x,y) \in E} J_*(x,y) \ln \frac{ p_{x,y} }{ p_{y,x} }\,.
\]
To conclude we observe that 
\[\frac{1}{2} \sum _{(x,y) \in E} J_*(x,y) \ln \frac{ p_{x,y} }{ p_{y,x} }
=\frac{1}{2} \sum_{k=1}^m \th_k \sum _{(x,y) \in E} J_k (x,y) \ln \frac{ p_{x,y} }{ p_{y,x} }
\]
and that, if $\cC_k = (z_1, \dots, z_s)$, 
\[
\frac{1}{2}  \sum _{(x,y) \in E} J_k (x,y) \ln \frac{ p_{x,y} }{ p_{y,x} }=
  \sum_{i=1}^s \frac{1}{2} \left[J_k(z_i, z_{i+1} )  \ln \frac{ p_{z_i,z_{i+1} } }{ p_{z_{i+1},z_i} }
  +J_k(z_{i+1}, z_{i} )  \ln \frac{ p_{z_{i+1},z_{i} } }{ p_{z_{i},z_{i+1}} }\right]= \cA( \cC_k)\,.
   \]

\medskip

 \end{proof}
 
  \section{Extended fluctuation theorem for the empirical  current of generic  semi--Markov processes }
  We conclude by discussing some extension  of the above analysis to generic semi--Markov processes.
  Again we assume \eqref{simone}, i.e.
  \[p_{x,y}>0 \text{ if and only if } p_{y,x}>0\,,\] 
   and we introduce the set $E$ according to \eqref{simone_bis}.

Recall the notation introduced in Subsection \ref{spada} and in particular the process 
  $\bbY:=(\bbY_t)_{t \geq 0}$ which is a DTI semi--Markov process with state space $E$. To $\bbY$  one can apply  Fact \ref{vettore}. On the other hand, $\bbY$ does  not belong to the range of application of Theorems \ref{chiave} and \ref{udine} since, given  states $x,y,z $ in $V$ with  $x \not =z$ and $(x,y), (y,z) \in E$, we have that $\hat p_{ (x,y) , (y,z) } >0$ but $\hat p_{(y,z), (x,y)}=0$.

 \smallskip

\begin{Proposition}\label{cocinella} Consider the LDP  rate functional $I (\mu, Q)$ of Fact \ref{vettore}  referred to the semi--Markov process $\bbY$  on $E$ with dynamical parameters $ \psi_{x,y}$ and $ \hat p _{\cdot, \cdot}$ In particular, $\mu \in \cP( E \times (0,+\infty])$ and $Q \in \bbR_+ ^{E \times E}$.    Let $I_*(\mu,Q)$ be the LDP  rate functional of Fact \ref{vettore}  referred to the semi--Markov process $\bbY_*$  on $E$ with dynamical parameters $ \psi^*_{x,y}:= \psi _{y,x} $ and $ \hat p _{\cdot, \cdot}$ Then 
\begin{equation}\label{hofame}
I( \mu, Q) = I_*(\mu_*, Q_*) -\frac{1}{2} \sum _{(x,y) \in E} [\mathcal{K}^Q (x,y) - \cK^Q(y,x)]\ln \frac{ p_{x,y} }{p_{y,x} } \,,
\end{equation}
where
\begin{align*}
&\mu_*( (x,y), d\t) := \mu ( (y,x) , d\t) \,, \\
& Q_* \bigl ( (x,y), (z,v) \bigr) := Q\bigl ( (v,z), (y,x) \bigr)\,,\\
&\cK ^Q (x,y) := \sum _{z \in V} Q( (x,y) , (y,z) ) \,.
\end{align*}
\end{Proposition}
\begin{proof}
 To simplify notation we write $\mu(xy, d\t)$ instead of $\mu ((x,y), d\t)$, $Q(xy, vz)$ instead of $Q(( x,y), (v,z) )$ and similarly  for $\mu_*$, $Q_*$. In general, we will write often $xy$ instead of $(x,y)$.
 
 Note that, given  $(x,y)\in E$ and  $ (v,z) \in E$, it holds  $\hat p_{(x,y), (v,z)}>0$ if and only if  $y=v$.  Assume  now that $Q( xy,vz) >0$ for some  $(x,y) , (v,z) \in E$ with  $y \not = v$. As a first consequence we get that   $I(\mu, Q)= \infty$ by Remark \ref{letto}. On the other hand, under the same assumption, we have $Q_*(zv, yx)= Q( xy,vz) >0$ and 
  $(z,v)\in E$, $ (y, x)\in E$, $y \not = v$, thus implying that $I_*(\mu_*, Q_*)=+\infty$ by the same arguments used above.  Hence,  under the above assumption, \eqref{hofame} is trivially satisfied. 
  
  \smallskip
  
  From now on we assume that $Q(xy, vz)=0$ for any  $(x,y) , (v,z) \in E$ with  $y \not = v$, and similarly  for $Q_*$.
  
  \smallskip
  
  We first claim that $(\mu, Q) \in \L_0  $ if and only if $(\mu, Q) \in \L_0^*$ ($\L_0^*$ being the analogous of $\L_0$ for the semi--Markov process $\bbY_*$  on $E$ with dynamical parameters $ \psi^*_{x,y}:= \psi _{y,x} $ and $ \hat p _{\cdot, \cdot}$). We prove the claim. 
Trivially $\mu(xy, d\t) \ll \psi_{x,y} (d\t) $ for any $(x,y) \in E$  if and  only if  $\mu^*(xy, d\t) \ll \psi^*_{x,y} (d\t) $ for any $(x,y) \in E$. Suppose that $(\mu, Q) \in \L_0  $. Then, by definition of $\L_0$, for any $(x,y) \in E$ it holds
\begin{equation}\label{notte1}
 Z_{xy}:= \int _{(0,+\infty) }\mu (xy, d\t) \frac{1}{\t} =\sum _z Q(xy,yz)= \sum _z Q(zx , xy)\,. \end{equation}
 Since  $Z^*_{yx}:=\int _{(0,+\infty) }\mu _* (yx , d\t) \frac{1}{\t}  = \int _{(0,+\infty) }\mu  (xy, d\t) \frac{1}{\t}=: Z_{xy}\,,$
the  above identity \eqref{notte1} can be rewritten  as
\begin{equation}\label{notte2}  Z^*_{yx}=\int _{(0,+\infty) }\mu _* (yx , d\t) \frac{1}{\t}  = \sum_z Q_*(zy,yx)= \sum_z Q_*( yx,xz) \,,\end{equation}
thus  completing the proof that $(\mu_*, Q_*) \in \L_0^*$. By the same arguments one gets that $(\mu, Q) \in \L_0$ if $(\mu_*, Q_*) \in \L_0^*$, concluding the derivation of the claim.
\medskip 
 
Due to the above claim we can restrict to the case $(\mu, Q) \in \L_0$ and $(\mu_*, Q_*) \in \L_0^*$ (otherwise, \eqref{hofame} reads $+\infty= +\infty$ which is trivially true).  

\smallskip
Since  $Z^*_{xy}= Z_{yx}$, we have $ \tilde \mu _*(xy, d \t)= \tilde \mu (yx, d \t)$ (recall the notation in \eqref{arrosticino1}). This implies that $H( \tilde \mu_*(xy, d \t) | \psi^*_{xy})=  H( \tilde \mu(yx, d\t) | \psi_{yx})$, and therefore  that 
\begin{equation}\label{pupazzo1}
\sum_{(x,y) \in E} Z^*_{xy} 
H( \tilde \mu_*(xy, d \t) | \psi^*_{xy})= \sum _{(y,x) \in E} Z_{yx} H( \tilde \mu(yx, d\t) | \psi_{yx})\,.
\end{equation}
Since moreover $\xi_{xy}^*= \xi_{yx}$ (recall \eqref{arrosticino2}) we have 
\begin{equation}\label{pupazzo2}
\sum _{(x,y)\in E} Z_{xy}^* \xi_{xy}^* \mu_* ( xy, \{+\infty\})= \sum _{(x,y) \in E} Z_{yx} \xi_{yx} \mu (yx, \{+\infty\})\,.
\end{equation}

Due to \eqref{pupazzo1}, \eqref{pupazzo2} and Fact \ref{vettore} we conclude that 
\begin{equation}\label{pupazzo3} 
\begin{split}
I(\mu,Q) & = I_* (\mu_*, Q_*) + \sum _{(x,y) \in E} Z_{xy} H( \hat p_{xy, \cdot }^Q | \hat p_{xy, \cdot} ) - 
\sum _{(x,y) \in E} Z^*_{xy} H( \hat p_{xy, \cdot }^{Q^*} | \hat p_{xy, \cdot} )\\
& = I_*(\mu_*, Q_*) - \sum_{(x,y,z) \in F} Q(xy,yz) \ln ( Z_{xy} p_{y,z}) +  \sum_{(x,y,z) \in F} Q^*(xy,yz) \ln ( Z^*_{xy} p_{y,z})
\end{split}
\end{equation}
where
\[F= \{ (x,y,z): (x,y) \in E, \; (y,z) \in E \}\,.
\]
Note that $ (x,y,z)\in F $   if and only if $(z,y,x)\in F$.

Recall \eqref{notte1} and \eqref{notte2}. They imply
\begin{align}
&\sum _{(x,y,z) \in F} Q(xy,yz) \ln Z_{xy}= \sum _{(x,y) \in E} Z_{xy} \ln Z_{xy}\,,\label{po1}\\
&\sum _{(x,y,z) \in F} Q(xy,yz) \ln p_{y,z}= \sum_{(y,z) \in E} Z_{yz} \ln p_{y,z}\,,\label{po2}\\
& \sum _{(x,y,z) \in F} Q_* (xy,yz) \ln Z^*_{xy}= \sum _{(z,y,x) \in F} Q(zy,yx) \ln Z_{yx}=
 \sum _{  (y,x) \in E} Z_{yx} \ln Z_{yx}\,, \label{po3}\\
 & \sum _{(x,y,z) \in F} Q^*(xy,yz) \ln p_{y,z}= \sum _{(z,y,x) \in F} Q(zy,yx) \ln p _{y,z} = \sum _{(z,y) \in E}Z_{zy} \ln p_{y,z}\,.\label{po4}
 \end{align}

Coming back to \eqref{pupazzo3} we get
\begin{multline*} I(\mu,Q) = I_* (\mu_*, Q_*) -\eqref{po1}- \eqref{po2} + \eqref{po3}+ \eqref{po4}= 
 I_* (\mu_*, Q_*) -  \eqref{po2} + \eqref{po4}\\= I_* (\mu_*, Q_*) - \sum _{(y,z) \in E} ( Z_{yz}- Z_{zy} ) \ln p_{y,z}\,.
\end{multline*}
To conclude it is enough to observe that 
\begin{equation*}
\begin{split}
 \sum _{(y,z) \in E} ( Z_{yz}- Z_{zy} ) \ln p_{y,z}& = \frac{1}{2}[ \sum _{(y,z) \in E}  ( Z_{yz}- Z_{zy} ) \ln p_{y,z}
+ \sum _{(y,z) \in E}  ( Z_{zy}- Z_{yz} ) \ln p_{z,y }]\\
& =\frac{1}{2}\sum _{(y,z) \in E}  ( Z_{yz}- Z_{zy} ) \ln \frac{p_{y,z}}{ p_{z,y}}\,, 
\end{split}
\end{equation*}
and observe that $\cK^Q(y,z)= Z_{yz}$, $\cK^Q(z,y)= Z_{zy}$ by \eqref{notte1}.
\end{proof}

\begin{Theorem} \label{corro} The empirical current $J_t$ of 
the generic semi--Markov process $\bbX$ on $V$  with dynamical parameters $p_{x,y}$ and $\psi_{x,y}$ satisfies a LDP with
speed $t$ and  
good rate function $\cI(J)$. Writing $\cI_*(J)$ for  the rate function obtained when replacing $\psi_{x,y}$ with $\psi^*_{x,y}:= \psi_{y,x}$, we have
\begin{equation}\label{esausta}
\cI (J) = \cI _*( -J) -\frac{1}{2} \sum _{(x,y) \in E} J  (x,y) \ln \frac{ p_{x,y} }{p_{y,x} } \,.
\end{equation}
\end{Theorem}
\begin{proof}
Let us write $Q^{\bbX} _t$ for the empirical flow associated to $(\bbX_t)_{t\geq0}$ and $Q_t^{\bbY}$ for the empirical flow associated to $(\bbY_t)_{t\geq0}$. We have  $Q_t^\bbX(x,y)= \sum _z Q_t ^\bbY(xy,yz) + O(1/t)$.  Hence,   we get the LDP of  $Q_t^\bbX$ with a good rate function   from the LDP of $Q_t^\bbY$ with a good rate function (the latter  holds by contraction due to Fact \ref{vettore} and since $\bbY$ is a DTI semi--Markov process). Since $J_t(x,y)= Q^\bbX_t(x,y)-Q^\bbX_t(y,x)$, by contraction we get that the LDP of $J_t$ with a good rate function. 

We observe now that 
\[ J_t (x,y)= Q^\bbX_t(x,y)-Q^\bbX_t(y,x) = \sum _z Q_t^\bbY( (x,y), (y,z))- \sum_z Q_t^\bbY( (z,y), (y,x) )+O(1/t)\,. \]
As a consequence of \eqref{hofame} and the above identity we get \eqref{esausta}.
\end{proof}

By combining Theorem \ref{corro} with \eqref{bu2} in Theorem \ref{udine} we have 
\begin{Corollary}\label{curioso}
In the same context of Theorem \ref{corro}, under the DTI  condition  (i.e. $\psi_{x,y}= \psi_x$ for all $x,y$), we have
\[\cI(J)= \cI_*(J) \,.\]
\end{Corollary}

\end{document}